%% file: rods.tex
\documentclass[12pt]{article}%

\input{prefix}
\input{notation}

\newcommand{\cn}{\mu}

\BibLatexMode{%
  \bibliography{master}%
}

\title{Cutting cycles of rods in space is FPT}
\author{
  Mitchell Jones\MitchellThanks{}%
}

\begin{document}

\maketitle

\begin{abstract}
  In this short note, we show that cutting cycles of rods is 
  fixed-parameter tractable by reducing the problem to computing a 
  feedback vertex set in a mixed graph.
\end{abstract}

\section{Introduction}

Let $S$ be a collection of $n$ non-vertical segments in $\Re^3$. For
two segments $s, s' \in S$, we say that $s \succ s'$ if and only if
$s$ \emph{lies above} $s'$. Formally, if $\ell$ is the vertical line
that meets $s$ and $s'$, then $s \succ s'$ if and only if the
$z$-coordinate of the point $\ell \cap s$ is larger than the
$z$-coordinate of the point $\ell \cap s'$. This induces a relation
$\succ$ on the segments in $S$ (which may not be transitive). In
general, this relation may contain what we call \emphi{depth cycles}.  
By \emph{cutting} the segments of $S$ into multiple smaller segments, we
obtain a new collection of segments $S'$. The goal is to cut the
segments of $S$ such that the relation for the segments $S'$ is a
partial ordering (i.e., there are no depth cycles).  The minimum
number of such cuts needed is called the \emphi{cutting number} of
$S$. See \figref{rods}. Often segments are referred to as rods, and we
use these terms interchangeably.

\begin{figure}[t]
  \centering%
  \includegraphics[page=1,scale=0.8]{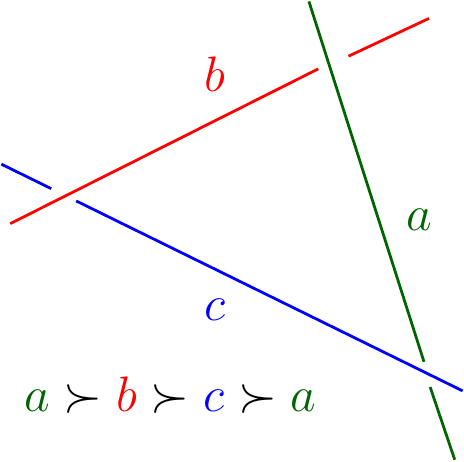}%
  \hspace{25pt}{\color{MJLightGrey}\vrule}\hspace{25pt}%
  \includegraphics[page=2,scale=0.8]{figs/rods-ex}%
  \caption{A set of segments inducing a depth cycle (left) and 
  the resulting depth ordering after making a single cut (right).}
  \figlab{rods}
\end{figure}

To summarize, the problem is to make as few cuts as possible to the 
rods (each rod may be cut in multiple different places) such that the
resulting depth ordering is acyclic. In this note, we
study the parameterized version of the problem.

\begin{problem}
  \problab{fpt:rods}
  Given a set $S$ of $n$ non-vertical segments in $\Re^3$ and a
  parameter $k$, is there an algorithm running in time $f(k) n^{O(1)}$
  (where $f$ is some computable function) which decides if all depth
  cycles can be eliminated using at most $k$ cuts?
\end{problem}

This is equivalent to asking if the above problem is 
\emphi{fixed-parameter tractable} ($\FPT$) by the solution size.
See \cite{fg-pct-06, cfk+-pa-15} for a reference on parameterized
algorithms and complexity.

\subsection{Previous work}
Cutting cycles of rods has applications in hidden surface removal and
computer graphics (see discussions in \cite{aks-ctcls-05} and
\cite[Chapter 12]{dbcvko-cgaa-08}).  The problem has been studied from
both a combinatorial \cite{aks-ctcls-05, as-edctd-18} and
computational \cite{s-ccrs-98, hs-oplpaa-01, adbgm-ccrsha-08}
viewpoint. Here, we focus on the computational aspect of the
problem.

For a given collection of segments $S$ in $\Re^3$ with
cutting number $\cn$, Solan \cite{s-ccrs-98} was the first
to develop an output sensitive algorithm for the problem, producing
a set of $O(n^{1+\eps}\cn^{1/2})$ cuts in time $O(n^{4/3+\eps}\cn^{1/3})$,
for any given $\eps > 0$. One can think of this as a 
$O(n^{1+\eps}\cn^{-1/2})$-approximation algorithm. This
was further improved to a 
$O(n\alpha(n)\cn^{-1/2} \log n)$-approximation\footnote{$\alpha(n)$ 
is the inverse Ackermann function.}
algorithm with expected running time $O(n^{4/3+\eps}\cn^{1/3})$ 
\cite{hs-oplpaa-01}. Aronov \etal~\cite{adbgm-ccrsha-08} obtain a 
$O(\log\cn\log\log\cn)$-approximation algorithm in 
$O(n^{4+2\omega}\log^2 n)$ time\footnote{Here, $\omega < 2.37$ 
is the matrix multiplication constant for multiplying two 
$n \times n$ matrices.}, 
by reducing the problem to computing a feedback vertex set in 
a directed graph.

As far as the author is aware, the problem of cutting cycles
of rods has not been studied from a parameterized perspective.

\paragraph*{Relation to feedback vertex set.}
Let $G = (V,A)$ be a directed graph with $n$ vertices $V$ and
arc set $A$. Recall that a subset of vertices $U \subseteq V$
form a \emphi{feedback vertex set} ($\FVS$) if $G - U$ is
acyclic. It is known that $\FVS \in \FPT$ when parameterized 
by the solution size $k$ \cite{cllor-fptdfvs-08}, running in
$O(4^k k! k^4 n^3)$ time.

In \cite{adbgm-ccrsha-08}, Aronov \etal~reduce the
problem to computing the feedback vertex set in a directed graph 
$G$. In particular, the authors show that
if all depth cycles in $S$ can be eliminated with $k$ cuts, then 
$G$ admits a FVS of size at most $2k$. Unfortunately the reverse
implication does not hold in their reduction. Since $\FVS \in \FPT$,
the best one can hope for is a 2-approximation algorithm
to \probref{fpt:rods} running in $O(4^k k! k^4 n^3)$ time, using 
their reduction and the result of Chen \etal \cite{cllor-fptdfvs-08}.

\subsection{Our results}

We prove that \probref{fpt:rods} is fixed-parameter tractable
by the solution size $k$.

\SaveContent{\ThmFPTRods}{
  Let $S$ be a collection of $n$ non-vertical segments in $\Re^3$.
  Given an integer parameter $k$, there is an algorithm which
  decides if all depth cycles of $S$ can be eliminated with at most
  $k$ cuts, and runs in time $O(47.5^k k! n^8)$.
}

\begin{theorem}
  \thmlab{fpt:rods}
  \ThmFPTRods{}
\end{theorem}

\section{Proof of \thmref{fpt:rods}}

\subsection{Mixed graphs}
A mixed graph $G$ is defined by the tuple $(V,E,A)$, where $V$
is a collection of $n$ vertices, $E$ is a collection of undirected
edges, and $A$ is a collection of directed arcs. For each edge
$e \in E$ with endpoints $u, v \in V$, we denote $e$ by $\set{u,v}$.
Each arc $a \in A$ with tail $u \in V$ and head $v \in V$ is 
denoted by $(u,v)$. The \emph{total degree} of a vertex $v$ is the
number of edges and arcs adjacent to $v$.
A \emph{cycle} in $G$ consists of a sequence
of distinct vertices $v_1, \ldots, v_\ell$ in $G$, where $v_\ell = v_1$
and for each $1 \leq i \leq \ell-1$ either $\set{v_i, v_{i+1}} \in E$
or $(v_i, v_{i+1}) \in A$. A subset of vertices $S \subseteq V$ is
a feedback vertex set for $G$ if and only if 
$G - S$ is acyclic. We will use the following result of 
\cite{bl-fvsmg-11}.

\begin{theorem}[\cite{bl-fvsmg-11}]
  \thmlab{fpt:fvs:mg}
  Let $G$ be a mixed graph on $n$ vertices. Given an integer parameter
  $k$, there is an algorithm which decides if $G$ contains a feedback
  vertex set of size at most $k$, and runs in time $O(47.5^k k! n^4)$.
\end{theorem}

\subsection{Reducing cutting cycles to $\FVS$ in mixed graphs}
\newcommand{\pr}[1]{\overline{#1}}

For a segment $s$ in $\Re^3$, let $\pr{s}$ be the projection
of $s$ onto the $xy$-plane. Similarly, let $\pr{S} = \Set{\pr{s}}{s \in S}$
for our collection of segments $S$. Given $s \in S$, let $I(s)$ be
the set of intersection points in the arrangement $\ArrX{\pr{S}}$ 
along the segment $\pr{s}$. We order these intersection points $I(s)$
along $\pr{s}$ arbitrarily and label them 
$p^s_1, \ldots, p^s_\ell \in \Re^2$, where $\ell = \cardin{I(s)}$.

Observe that by cutting each segment $s$ at each of its intersection
points (with the appropriate $z$-coordinate in $\Re^3$), we remove all 
depth cycles (and make $O(n^2)$ cuts). In particular,
it suffices to focus on cutting segments \emph{only} at intersection
points in order to remove depth cycles. Indeed, any solution which cuts 
a segment $s$ at a point which is a non-intersection can be shifted
along $s$ until it meets an intersection, and remains a valid solution.

Equipped with these observations, we define the following mixed 
graph $G_S = (V,E,A)$ (see also \figref{cuts}):
\begin{compactenumi}
  \item for each $s \in S$, 
  define the vertex set $V_s = \{v^s_1, \ldots, v^s_{\cardin{I(s)}}\}$, where
  the vertex $v^s_i$ represents the $i$th intersection point
  $p^s_i$ in $I(s)$. We set $V = \bigcup_{s \in S} V_s$;
  
  \item for each $s \in S$, 
  add the undirected edges $\set{v^s_i, v^s_{i+1}}$ 
  for $1 \leq i \leq \cardin{I(s)}-1$ to $E$;
  
  \item for two segments $s \succ t$ for which 
  $\pr{s} \cap \pr{t} \neq \varnothing$, 
  suppose that $p^s_i = p^t_j$ for some integers
  $i, j$ (i.e., $\pr{s} \cap \pr{t}$ is the $i$th intersection point
  along $\pr{s}$, and $\pr{s} \cap \pr{t}$ is the $j$th intersection point
  along $\pr{t}$). Then we add the directed arc $(v^s_i, v^t_j)$ to $A$.
\end{compactenumi}

\begin{figure}
  \centerline{
    \centering        
    \includegraphics[page=3,scale=0.8]{figs/rods-ex}%
    \hspace{25pt}{\color{MJLightGrey}\vrule}\hspace{25pt}%
    \includegraphics[page=1,scale=0.6]{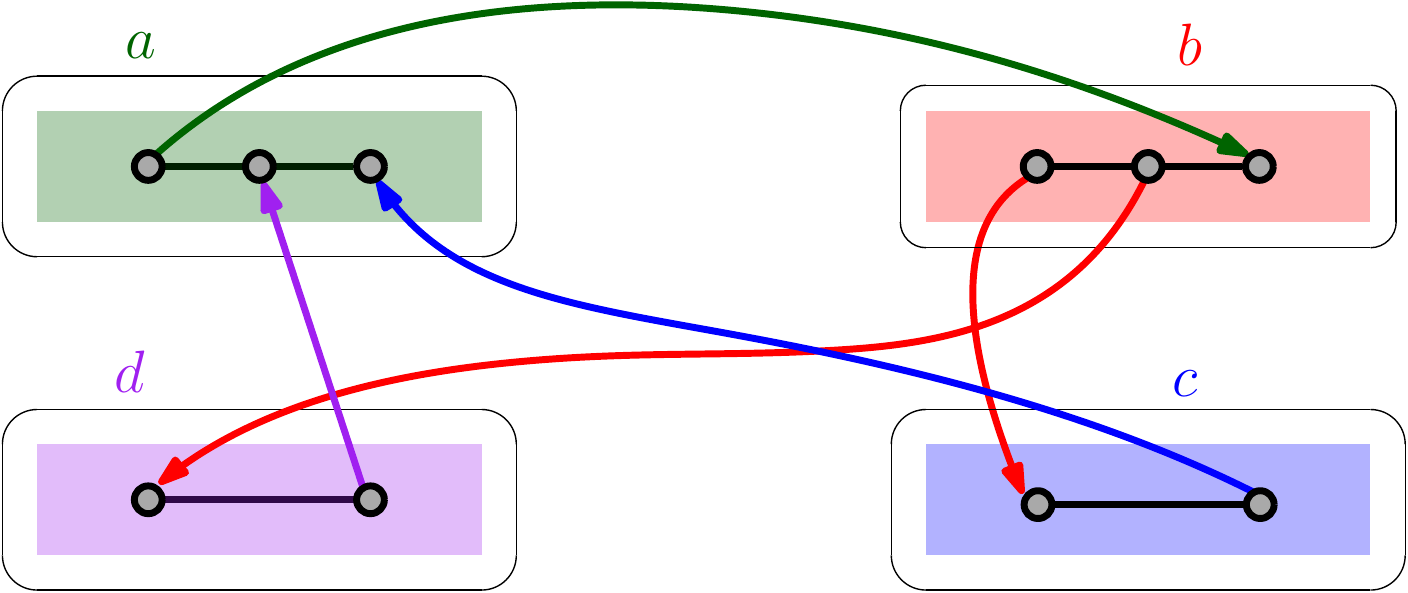}%
  }
  \caption{A set of four segments $S$ and the
  corresponding mixed graph $G_S$.}
  \figlab{cuts}%
\end{figure}

\begin{lemma}
  \lemlab{cuts:fvs}
  All cycles in $S$ can be eliminated with at most $k$ cuts if
  and only if the mixed graph $G_S$ admits a feedback vertex set
  of size at most $k$.
\end{lemma}
\begin{proof}
  Let $U \subseteq V$ be a feedback vertex set for $G_S$.
  Since each vertex in $U$ corresponds to an intersection
  point $p^s_i$ along some segment $s \in S$, we cut $s$ at the
  point with $xy$-coordinates $p^s_i$ and the appropriate 
  $z$-coordinate. Let $C$ denote the resulting collection of cuts. 
  Clearly, the number of cuts we make is $\cardin{C} = \cardin{U} \leq k$. 
  We claim that $C$ eliminates all depth cycles. Suppose not, then
  there is a cycle consisting of (potentially cut) segments
  $s_1 \succ s_2 \succ \ldots \succ s_p$, where
  $s_p = s_1$. Since $s_i \succ s_{i+1}$ for $1 \leq i \leq p-1$,
  there is a directed edge $(v^{s_i}_{j(i)}, v^{s_{i+1}}_{j'(i+1)})$
  for each $i$ and indices $1 \leq j(i), j'(i) \leq \cardin{I(s_i)}$.
  Here, the index $j(i)$ (resp.~$j'(i)$) is chosen such that
  $\pr{s_i} \cap \pr{s_{i+1}}$ (resp.~$\pr{s_i} \cap \pr{s_{i-1}}$) 
  is the $j(i)$th (resp.~$j'(i)$th) intersection point
  along $\pr{s_i}$.
  Since the vertices of $V_{s_{i+1}}$ are connected by an undirected path, 
  we can travel from the vertex $v^{s_{i+1}}_{j'(i+1)}$ to the vertex 
  $v^{s_{i+1}}_{j(i+1)}$ along an undirected path $P_i$. 
  This corresponds to moving along the segment
  $s_{i+1}$. Note that 
  since each vertex in $V$ is adjacent to exactly one arc, we have that
  $v^{s_{i}}_{j'(i)} \neq v^{s_{i}}_{j(i)}$ for all $i$. 
  Next, we can take another directed edge 
  $(v^{s_{i+1}}_{j(i+1)}, v^{s_{i+2}}_{j'(i+2)})$. Continuing in 
  this fashion, we can stitch together a cycle which concatenates
  the directed edge $(v^{s_i}_{j(i)}, v^{s_{i+1}}_{j'(i+1)})$
  followed by the undirected path $P_{i+1}$ from $v^{s_{i+1}}_{j'(i+1)}$
  to $v^{s_{i+1}}_{j(i+1)}$, for $i = 1, \ldots, p-1$,
  to obtain a cycle in $G_S - U$. A contradiction.
  
  Conversely, let $C$ be a collection of cuts which eliminates
  all depth cycles in $S$. By the above discussion, without loss
  of generality we can assume that $C$ cuts the segments of $S$
  at intersection points. Since each intersection point is associated
  to a vertex in $G_S$, we add the vertex $v^s_i$ to $U$ if and 
  only if segment $s$ was cut at the intersection point
  $p^s_i$. Clearly, this collection of vertices $U$ satisfies
  $\cardin{U} = \cardin{C} \leq k$. The argument that
  $U$ is a feedback vertex set for $G_S$ is similar to above. If
  $U$ is a not a solution, then there is a cycle in $G_S$
  which corresponds to a depth cycle in $S$ after making the
  cuts in $C$, a contradiction.
\end{proof}

\RestatementOf{\thmref{fpt:rods}}{\ThmFPTRods}
\begin{proof}
  The algorithm is straightforward. Given $S$, apply the above
  reduction to obtain a mixed graph $G_S$. Note that $G_S$ has
  size $O(n^2)$ (there are $O(n^2)$ vertices, and each vertex has
  total degree at most three), and can be easily computed in $O(n^2)$ 
  time.
  Next, run the algorithm of \thmref{fpt:fvs:mg} on $G_S$ to
  obtain a feedback vertex set $U$. One can map this set $U$ back
  to a cut set $C$ for $S$ by \lemref{cuts:fvs}. The total running time
  of the algorithm is 
  $O(47.5^k k! \cardin{G_S}^4) = O(47.5^k k! n^8)$.
\end{proof}


\paragraph*{Acknowledgements.}
Thanks to Sariel Har-Peled for suggesting the problem and subsequent
discussions. Thanks also to Chandra Chekuri for discussions on feedback 
vertex set and its variants.


\BibTexMode{%
   \bibliographystyle{alpha}%
   \bibliography{master}%
}%

\BibLatexMode{\printbibliography}

\end{document}

%% file: prefix.tex
\providecommand{\BibLatexMode}[1]{}
\providecommand{\BibTexMode}[1]{#1}

\ifx\UseBibLatex\undefined%
  \renewcommand{\BibLatexMode}[1]{}
  \renewcommand{\BibTexMode}[1]{#1}
\else
  \renewcommand{\BibLatexMode}[1]{#1}
  \renewcommand{\BibTexMode}[1]{}
\fi

\BibLatexMode{%
   \usepackage[bibencoding=utf8,style=alphabetic,%
   backend=biber,sortlocale=en_US]{biblatex}%
   \usepackage[english]{babel}
   \usepackage{csquotes}
   \usepackage{master}%
}

\usepackage[cm]{fullpage}
\usepackage{url, graphicx}%
\usepackage{amsmath}%
\usepackage{amssymb}%
\usepackage{amsbsy}%

\usepackage[full,langfont=caps,funcfont=roman]{complexity}
\usepackage{mleftright}

\newlang{\FVS}{FVS}
\newlang{\PFVS}{PairedFVS}
\newlang{\VtC}{VC}
\newlang{\TSAT}{3\text{-}SAT}

\usepackage{wasysym}%
\usepackage{picins}%
\usepackage{xcolor}
\usepackage{xspace}%
\usepackage{marvosym}%
\usepackage{euscript}%
\usepackage{caption}%
\usepackage{scalerel}
\usepackage{mathcommand}

\usepackage[short,nocomma]{optidef}

\usepackage{footnote}
\usepackage[stable]{footmisc}

\RequirePackage[charter]{mathdesign}
\RequirePackage[scaled]{berasans}
\RequirePackage[semibold]{sourcecodepro}
\RequirePackage[scaled=.96,osf,sups]{XCharter}

\SetMathAlphabet{\mathsf}{bold}{\encodingdefault}{\sfdefault}{b}{\updefault}
\SetMathAlphabet{\mathtt}{bold}{\encodingdefault}{\ttdefault}{b}{\updefault}
\SetMathAlphabet{\mathsf}{normal}{\encodingdefault}{\sfdefault}{\mddefault}{\updefault}
\SetMathAlphabet{\mathtt}{normal}{\encodingdefault}{\ttdefault}{\mddefault}{\updefault}

\RequirePackage{stmaryrd,textcomp}
\RequirePackage{microtype}

\usepackage{mathcalb}%

\newlength{\savedparindent}

\newcommand{\SaveContent}[2]{%
   \expandafter\newcommand{#1}{#2}%
}

\newcommand{\RestatementOf}[2]{
   \noindent%
   \textbf{Restatement of #1.}
   {\em #2{}}%
}

\definecolor[named]{MJBlue}{cmyk}{1,0.1,0,0.1}
\definecolor[named]{MJYellow}{cmyk}{0,0.16,1,0}
\definecolor[named]{MJOrange}{cmyk}{0,0.42,1,0.01}
\definecolor[named]{MJRed}{HTML}{903838}
\definecolor[named]{MJLightBlue}{cmyk}{0.49,0.01,0,0}
\definecolor[named]{MJGreen}{HTML}{248a6d}
\definecolor[named]{MJPurple}{cmyk}{0.55,1,0,0.15}
\definecolor[named]{MJDarkBlue}{cmyk}{1,0.58,0,0.21}
\definecolor[named]{MJLightGrey}{HTML}{829197}

\usepackage{hyperref}%
\hypersetup{%
  breaklinks,%
  colorlinks=true,%
  urlcolor=MJGreen,
  linkcolor=MJDarkBlue,
  citecolor=MJPurple
}
   
\usepackage{titlesec}
\titlelabel{\thetitle. }

\usepackage[inline]{enumitem}

\newlist{compactenumA}{enumerate}{5}%
\setlist[compactenumA]{topsep=0pt,itemsep=-1ex,partopsep=1ex,parsep=1ex,%
   label=(\Alph*)}%

\newlist{compactenuma}{enumerate}{5}%
\setlist[compactenuma]{topsep=0pt,itemsep=-1ex,partopsep=1ex,parsep=1ex,%
   label=(\alph*)}%

\newlist{compactenumI}{enumerate}{5}%
\setlist[compactenumI]{topsep=0pt,itemsep=-1ex,partopsep=1ex,parsep=1ex,%
   label=(\Roman*)}%

\newlist{compactenumi}{enumerate}{5}%
\setlist[compactenumi]{topsep=0pt,itemsep=-1ex,partopsep=1ex,parsep=1ex,%
   label=(\roman*)}%

\newlist{compactitem}{itemize}{5}%
\setlist[compactitem]{topsep=0pt,itemsep=-1ex,partopsep=1ex,parsep=1ex}%

\newlist{ProblemList}{enumerate}{5}%
\setlist[ProblemList]{topsep=5pt,%
   partopsep=0.03ex,parsep=0.365ex,%
   resume=problemlist,leftmargin=0.6in,%
   label=\bf\color{MJRed}{Prob. \arabic*}:, ref=\arabic*}

\newcommand{\HLinkShort}[2]{\hyperref[#2]{#1\ref*{#2}}}
\newcommand{\HLink}[2]{\hyperref[#2]{#1~\ref*{#2}}}
\newcommand{\HLinkPage}[2]{\hyperref[#2]{#1~\ref*{#2}%
      $_\text{p\pageref{#2}}$}}
\newcommand{\HLinkPageOnly}[1]{\hyperref[#1]{Page~\refpage*{#1}%
      $_\text{p\pageref{#1}}$}}

\newcommand{\HLinkSuffix}[3]{\hyperref[#2]{#1\ref*{#2}{#3}}}
\newcommand{\HLinkPageSuffix}[3]{\hyperref[#2]{#1\ref*{#2}%
      #3$_\text{p\pageref{#2}}$}}

\newcommand{\figlab}[1]{\label{fig:#1}}
\newcommand{\figref}[1]{\HLink{Figure}{fig:#1}}

\newcommand{\defrefY}[2]{\hyperref[def:#2]{#1}}

\newcommand{\lemlab}[1]{\label{lemma:#1}}
\newcommand{\lemref}[1]{\HLink{Lemma}{lemma:#1}}%

\newcommand{\problab}[1]{\label{prob:#1}}%
\newcommand{\probref}[1]{\HLink{Problem}{prob:#1}}%

\newcommand{\thmlab}[1]{{\label{theo:#1}}}
\newcommand{\thmref}[1]{\HLink{Theorem}{theo:#1}}

\providecommand{\eqlab}[1]{}%
\renewcommand{\eqlab}[1]{\label{equation:#1}}

\newcommand{\lpref}[1]{(\hyperref[lp:#1]{\ref*{lp:#1}})}

\usepackage[amsmath,thmmarks,framed]{ntheorem}%
\theoremseparator{.}%

\theoremstyle{plain}%
\newtheorem{theorem}{Theorem}

\newtheorem{lemma}[theorem]{Lemma}

\theoremstyle{plain}%
\theoremheaderfont{\sf}
\theorembodyfont{\upshape}%
\newtheorem*{remark:unnumbered}[FakeCounter]{Remark}%

\newtheorem*{defn:unnumbered}[FakeCounter]{Definition}

\newtheorem{problem}{Problem}

\newcommand{\myqedsymbol}{\rule{2mm}{2mm}}

\theoremheaderfont{\em}%
\theorembodyfont{\upshape}%
\theoremstyle{nonumberplain}%
\theoremseparator{}%
\theoremsymbol{\myqedsymbol}%
\newtheorem{proof}{Proof:}%

\newcommand{\emphic}[2]{%
   \textcolor{MJRed}{%
      \textbf{\emph{#1}}}%
   \index{#2}}

\newcommand{\emphi}[1]{\emphic{#1}{#1}}

\newcommand{\atgen}{\symbol{'100}}

\newcommand{\MitchellThanks}[1]{%
   \thanks{%
      Department of Computer Science;
      University of Illinois; 201 N. Goodwin Avenue; Urbana, IL,
      61801, USA; {\tt mfjones2\atgen{}illinois.edu}; {\tt
         \url{http://mfjones2.web.engr.illinois.edu/}.} #1}}

%% file: notation.tex
\newcommand{\pth}[1]{\mleft({#1}\mright)}

\newcommand{\set}[1]{\mleft\{#1\mright\}}
\newcommand{\Set}[2]{\mleft\{{#1}\mid{#2}\mright\}}

\newcommand{\etal}{{e{}t~a{}l.}\xspace}

\newcommand{\cardin}[1]{\left\vert {#1}  \right\vert}

\newcommand{\eps}{\varepsilon}

\LoopCommands\lettersLowercase[v#1]{\newcommand#2{\boldsymbol#1}}

\LoopCommands{ABCDEFGIJKLMNOPQRTVWXYZ}[#1S]{\newcommand#2{#1}}

\LoopCommands{ACFGINPQZ}[#1e]{\newmathcommand#2{\mathbb#1}}
\LoopCommands{R}[#1e]{\renewmathcommand#2{\mathbb#1}}

\newcommand{\ArrX}[1]{\mathcal{A}\pth{#1}}

%% file: rods.bbl
\newcommand{\etalchar}[1]{$^{#1}$}
\begin{thebibliography}{dBCvKO08}

\bibitem[AdBGM08]{adbgm-ccrsha-08}
Boris Aronov, Mark de~Berg, Chris Gray, and Elena Mumford.
\newblock Cutting cycles of rods in space: hardness and approximation.
\newblock In Shang{-}Hua Teng, editor, {\em Proc. 19th ACM-SIAM Sympos.
  Discrete Alg. {\em(SODA)}}, pages 1241--1248. {SIAM}, 2008.

\bibitem[AKS05]{aks-ctcls-05}
Boris Aronov, Vladlen Koltun, and Micha Sharir.
\newblock Cutting triangular cycles of lines in space.
\newblock {\em Discrete Comput. Geom.}, 33(2):231--247, 2005.

\bibitem[AS18]{as-edctd-18}
Boris Aronov and Micha Sharir.
\newblock Almost tight bounds for eliminating depth cycles in three dimensions.
\newblock {\em Discrete Comput. Geom.}, 59(3):725--741, 2018.

\bibitem[BL11]{bl-fvsmg-11}
Paul~S. Bonsma and Daniel Lokshtanov.
\newblock Feedback vertex set in mixed graphs.
\newblock In {\em Proc. 12th Workshop Alg. Data Struct. {\em(WADS)}}, volume
  6844 of {\em Lecture Notes in Computer Science}, pages 122--133. Springer,
  2011.

\bibitem[CFK{\etalchar{+}}15]{cfk+-pa-15}
Marek Cygan, Fedor~V. Fomin, Lukasz Kowalik, Daniel Lokshtanov, D{\'{a}}niel
  Marx, Marcin Pilipczuk, Michal Pilipczuk, and Saket Saurabh.
\newblock {\em Parameterized Algorithms}.
\newblock Springer, 2015.

\bibitem[CLL{\etalchar{+}}08]{cllor-fptdfvs-08}
Jianer Chen, Yang Liu, Songjian Lu, Barry O'Sullivan, and Igor Razgon.
\newblock A fixed-parameter algorithm for the directed feedback vertex set
  problem.
\newblock {\em J. {ACM}}, 55(5):21:1--21:19, 2008.

\bibitem[dBCvKO08]{dbcvko-cgaa-08}
Mark de~Berg, Otfried Cheong, Marc~J. van Kreveld, and Mark~H. Overmars.
\newblock {\em Computational geometry: algorithms and applications, 3rd
  Edition}.
\newblock Springer, 2008.

\bibitem[FG06]{fg-pct-06}
J{\"{o}}rg Flum and Martin Grohe.
\newblock {\em Parameterized Complexity Theory}.
\newblock Texts in Theoretical Computer Science. An {EATCS} Series. Springer,
  2006.

\bibitem[HS01]{hs-oplpaa-01}
Sariel Har{-}Peled and Micha Sharir.
\newblock Online point location in planar arrangements and its applications.
\newblock {\em Discrete Comput. Geom.}, 26(1):19--40, 2001.

\bibitem[Sol98]{s-ccrs-98}
Alexandra Solan.
\newblock Cutting cylces of rods in space.
\newblock In {\em Proc. 14th Int. Annu. Sympos. Comput. Geom. {\em (SoCG)}},
  pages 135--142, 1998.

\end{thebibliography}
